\documentclass[creativecommons]{eptcs}
\providecommand{\event}{V2CPS 2016} 

\providecommand{\doi}[1]{\textsc{doi}: \href{http://dx.doi.org/#1}{\nolinkurl{#1}}}

\usepackage{graphicx}
\usepackage{amssymb}
\usepackage{amsmath}
\usepackage{mathrsfs}
\usepackage{gastex}
\usepackage{paralist}
\usepackage{pgf}
\usepackage{tikz}
\usepackage{url}

\usetikzlibrary{automata}

\newcommand{\encode}{\mbox{enc}}

\newcommand{\set}[1]{\left\{ #1 \right\}}

\newcommand{\Nat}{\mathbb N}
\newcommand{\seq}[1]{\langle #1 \rangle}

\newcommand{\Real}{\mathbb R}
\newcommand{\Rplus}{{\mathbb R}^{+}}
\newcommand{\OptCost}{\mathsf{OptCost}}
\newcommand{\Runs}{\mathsf{Runs}}
\newcommand{\pricevar}{\mbox{price}}

\def\rmdef{\stackrel{\mbox{\rm {\tiny def}}}{=}} 
\renewcommand{\triangleq}{\rmdef}

\newtheorem{theorem}{Theorem}[section]
\newtheorem{lemma}[theorem]{Lemma}

\newenvironment{proof}[1][Proof]{\begin{trivlist}
\item[\hskip \labelsep {\bfseries #1}]}{\end{trivlist}}
\newenvironment{definition}[1][Definition]{\begin{trivlist}
\item[\hskip \labelsep {\bfseries #1}]}{\end{trivlist}}

\newcommand{\qed}{\nobreak \ifvmode \relax \else
      \ifdim\lastskip<1.5em \hskip-\lastskip
      \hskip1.5em plus0em minus0.5em \fi \nobreak
      \vrule height0.75em width0.5em depth0.25em\fi}

\begin{document}

\title{On Nonlinear Prices in Timed Automata}

\author{Devendra Bhave  
\institute{IIT Bombay, India}
\email{devendra@cse.iitb.ac.in}
\and Shankara Narayanan Krishna 
\institute{IIT Bombay, India}
\email{krishnas@cse.iitb.ac.in}
\and Ashutosh Trivedi
\thanks{This material is based upon work supported by DARPA under agreement
  number FA8750-15-2-0096 and by the US National Science Foundation (NSF) under
  grant numbers CPS-1446900. The U.S. Government is authorized to reproduce and
distribute reprints for Governmental purposes notwithstanding any copyright
notation thereon. All opinions expressed are those of the authors, and not
necessarily of the DARPA or NSF.} 
\institute{CU Boulder, USA}
\email{ashutosh.trivedi@colorado.edu}
}

\def\titlerunning{On Nonlinear Prices in Timed Automata}
\def\authorrunning{D. Bhave, S. N. Krishna \& A. Trivedi}

\maketitle

\begin{abstract}
  Priced timed automata provide a natural model for quantitative analysis of
  real-time systems and have been successfully applied in various scheduling and
  planning problems. 
  The optimal reachability problem for linearly-priced timed automata is known
  to be PSPACE-complete. 
  In this paper we investigate priced timed automata with more general
  prices and show that in the most general setting the optimal reachability
  problem is undecidable. 
  We adapt and implement the construction of
  Audemard, Cimatti, Kornilowicz, and Sebastiani
  for non-linear priced timed automata using state-of-the-art
  theorem prover Z3 and present some preliminary results.
\end{abstract} 

\section{Introduction}
Timed automata, introduced by Alur and Dill~\cite{journals/tcs/AlurD94},
extend finite state automata with continuous variables---referred as
clocks---that evolve with uniform rates.  
Time automata syntax permits comparing clocks with integers as guard on
transitions and as well as invariants on locations (states),  and also allows
clock resets as a way to remember the time a transition was last fired. 
These features of time automata are general enough to permit modeling rich timing
properties of real-time systems while providing a decidable verification
framework.
Timed automata have been quite successful in practice due to their appealing
theoretical properties as well as the presence of mature verification tools such
as UPPAAL.  

Priced timed automata
\cite{DBLP:conf/hybrid/BouyerBL04,Bouyer:2007:ORP:1288667.1288679} are
extensions of timed automata which permit us to  model cost associated with
staying at locations as well as taking discrete transitions. 
Priced timed automata are useful in modeling various decision-theoretic problem
in the presence of strict timing constraints. 
The most natural problem studied on these models is the optimal reachability problem
 (shortest path problem) where the goal is to find the minimum (or maximum) cost
to reach a given set of locations. 

Linearly-priced timed automata~\cite{BGA2001} (LPTA), also known as weighted
timed automata, are subclasses of priced timed automata where  prices change
linearly with respect to delay incurred at particular location.
For LPTA the optimal reachability problem is known to be decidable and is
shown to be PSPACE-complete exploiting a clever extension of region graphs to
so-called corner-point abstraction by Bouyer et
al.~\cite{Bouyer:2007:ORP:1288667.1288679}.  
Alur et al. \cite{Alur:2004:OPW:1026907.1026909} earlier gave an EXPTIME
algorithm to solve the problem with an arbitrary initial state by giving a
non-trivial extension of the region graph. 
Larsen et al. \cite{BGA2001,DBLP:conf/hybrid/BouyerBL04} gave a symbolic
algorithm to solve the problem, although with some restrictions on the initial
state (a corner state with all clocks set to zero).  
A recent result by Fearnley and Jurdzinksi~\cite{FJ13} showed that the
PSPACE-hardness results hold for timed automata with two clocks~\cite{FJ13}. 
On the other hand, for timed automata with one clock, reachability-time
and reachability-price problems are known to be NL-complete~\cite{LMS04}.

In practice, however, the requirement for nonlinear pricing models is quite
common.  
As an example consider the optimal scheduling problem of battery usage in
embedded systems studied by Jongerden et al.~\cite{5491259}.
In this work authors modeled batteries using kinetic battery model (KiBaM). 
KiBaM itself is a nonlinear model, but Jongerden et al.~\cite{5491259} manually
discretized it to required approximation to model the whole problem as
optimization on LPTA. 
Similar scenarios can be cited from other application domains of priced timed
automata such as scheduling~\cite{behrmann2005optimal}, resource modeling and
analysis~\cite{ivanov2010remes}, and optimal synthesis~\cite{jiang2012modeling}.
However, we believe that providing non-linear price modeling facilities directly
in the language of timed automata will further their applicability in system
design. 
Jurdzinski and Trivedi~\cite{DBLP:conf/formats/JurdzinskiT08} introduced a 
non-linear subclass of priced timed automata, so-called concavely-priced timed
automata,  where prices in each location are certain concave prices of valuation
and time delays.  
Exploiting the concave nature of the prices, they showed that the optimal price
reachability problem for this class of automata has the same complexity as that of
LPTA. 
Priced timed automata with exponential price functions were studied in a
restricted context by Bouyer et al. \cite{Bouyer:2010:TAO:1755952.1755963} 
and Fahrenberg and Larsen \cite{Fahrenberg2009179}. 

In this paper we uniformly study various subclasses of (non-linear) priced timed
automata,  and study the boundary between decidable and undecidable
variants of PTA. 
Towards this goal we first show the undecidability of the optimal reachability
problem for unrestricted priced timed automata by showing a reduction from the
halting problem for two-counter machines.
For reasoning with decidable variants, we first introduce a key notion of
price-preserving bisimilarity.
We exploit this notion to formalize reduction for the optimal cost reachability
problem for piecewise-linear priced timed automata to linearly priced timed
automata. 
We also show the decidability of $\varepsilon$-optimal cost reachability for
priced timed automata with Lipschitz-continuous prices.  
Finally, we adapt the construction of 
Audemard, Cimatti, Kornilowicz, and Sebastiani~\cite{bmc-ta} 
for bounded model-checking of timed automata using SAT solvers to
work for bounded reachability problem for non-linearly priced timed automata using
SMT solver Z3~\cite{de2008z3}. 
In conjunction with a decision procedure for the theory of the class of  price
functions (for instance polynomial prices~\cite{jovanovic2012solving,de2008z3}),
our implementation can be used to compute bounded-step cost-optimal schedules
for priced timed automata. 
We demonstrate the applicability of our approach using airplane landing
problem~\cite{behrmann2005optimal}. 

This paper is organized as follows: we begin by defining syntax and semantics 
of generalized priced timed automata in the next section. 
We define various pricing models and their hierarchy. We prove key
undecidability result in Section~\ref{sec:undec} and show decidability results
in Section~\ref{sec:dec}. 
Finally, in section~\ref{sec:tool} we present the details of our implementation
and experimental results.

\section{Priced Timed Automata}
We denote sets of integers, rational numbers and real numbers as $\mathbb{Z},
\mathbb{Q}$ and $\mathbb{R}$ respectively. Their respective non-negative
subsets are denoted as 
$\mathbb{Z}^+, \mathbb{Q}^+$ and $\mathbb{R}^+$.

Let $X = \{x_1, x_2, \ldots,
x_n\}$ be the finite set of clocks. A clock valuation is a map $\nu: X \mapsto
\mathbb{R}^+$. Thus, a given clock valuation  
$\nu$ maps clock $x_i$ to a value $\nu_i$. This fact is written as $\nu(x_i) =
\nu_i$. In $n$-tuple form, a clock valuation $\nu$ is denoted as $(\nu_1,
\nu_2, \ldots, \nu_n)$. Given a clock valuation  
$\nu$ and $\tau \in \mathbb{R}^+$, $\nu + \tau$ is the clock valuation defined by
$(\nu_1 + \tau, \nu_2 + \tau, \ldots, \nu_n + \tau)$. A guard is any finite
conjunction of clauses of the form $x_i \sim c$, 
where clock $x_i \in X$, constant $c \in \mathbb{Z}^+$ and $\sim$ is one of the
comparison operators in set $\{<, \leq, =, >, \geq\}$. Let $G$ be the set of
guards. Given a valuation $\nu$ and a guard $g = 
\bigwedge_{j} (x_i \sim c_j)$, $\nu \vDash g$ means expression $\bigwedge_{j}
(\nu(x_i) \sim c_j)$ evaluates to true. For $Y \subseteq X$, $\nu[Y := 0]$ denotes
clock valuation in which clocks in $Y$ are  reset to 0 while other clocks remain
unchanged. 

\paragraph{\bf Timed Automata.}
A timed transition system $\mathcal{T}$ is a tuple $(L, X, E)$ where 
\begin{inparaenum}[(i)]
\item $L$ is a finite set of locations, 
\item $X$ is a finite set of clocks variables, and 
\item $E$ is set of transitions.
\end{inparaenum}
A \emph{configuration} of $\mathcal{T}$ is a pair $(\mathcal{\ell}, \nu)$, where
$\mathcal{\ell} \in L$ is a location and $\nu$ in clock valuation over set $X$. 
Let $Q_{\mathcal{T}}$ be the set of configurations for the timed transition system
$\mathcal{T}$. 
There are two types of transitions over $Q_{\mathcal{T}}$:
\begin{itemize}
\item 
  Delay, $E^\tau \subseteq Q_{\mathcal{T}} \times \mathbb{R}^+
  \times Q_{\mathcal{T}}$: $(\mathcal{\ell}, \nu) \xrightarrow{t}
  (\mathcal{\ell}, \nu+t)$, where $t \in \mathbb{R}^+$ 
\item 
  Switch, $E^e \subseteq Q_{\mathcal{T}} \times 2^X \times
  Q_{\mathcal{T}}$: $(\mathcal{\ell}, \nu) \xrightarrow{Y} (\mathcal{\ell}',
  \nu[Y:=0])$ where $Y \subseteq X$.
\end{itemize}
We write $E = E^{\mathcal{\tau}} \cup E^{e}$ for the set of transitions of timed
transition system $\mathcal{T}$. 

\begin{definition}
 A timed automaton $\mathcal{A}$ is a tuple $(L, X, E, I)$ where 
 \begin{inparaenum}[(i)]
  \item $L$ is a finite  set of locations
  \item $X$ is a finite set of clocks
  \item $E \subseteq L \times G \times 2^X \times L$ is a finite set of edges
  \item $I: L \mapsto G$ assigns an invariant to each location.
 \end{inparaenum} 
\end{definition}
 The semantics of timed automaton $\mathcal{A}$ is given as a timed transition
 system $\mathcal{T}_\mathcal{A} = (L_{\mathcal{A}}, X_{\mathcal{A}},
 E_{\mathcal{A}})$, where  
 \begin{inparaenum}[(i)]
  \item $L_{\mathcal{A}} = L$
  \item $X_{\mathcal{A}} = X$
  \item $E_{\mathcal{A}} = E^{\tau}_{\mathcal{A}} \cup E^{e}_{\mathcal{A}}$, s.t.
 \end{inparaenum}
 \begin{itemize}
  \item $E^{\tau}_{\mathcal{A}} = \{(\mathcal{\ell}, \nu) \xrightarrow{t} (\mathcal{\ell}, \nu+t) \; | \; t \in \mathbb{R}^+ \; \mbox{and} \; \forall \delta \in \mathbb{R}^+, 0 \leq \delta \leq t \Rightarrow (\nu+\delta)
  \vDash I(\mathcal{\ell})\}$
  \item $E^{e}_{\mathcal{A}} = \{ (\mathcal{\ell}, \nu) \xrightarrow{Y} (\mathcal{\ell}', \nu[Y{:=0}]) \; | \; Y {\subseteq} X, \; (\mathcal{\ell}, g, Y, \mathcal{\ell}') \subseteq E, \; \nu \vDash g \; \mbox{and} \; \nu \vDash I(\mathcal{\ell}) \}$
 \end{itemize}

A \emph{run} $\rho = q_0 \rightarrow q_1 \rightarrow \cdots \rightarrow q_m$ of the
timed automaton $\mathcal{A}$ is a finite path in the induced timed transition
system $\mathcal{T}_\mathcal{A}$  where every $q_i$ is configuration in
$\mathcal{T}_\mathcal{A}$ and $\rightarrow$ is either delay or switch edge in
$E_\mathcal{A}$. 
We use notation $\rho = q_0 \rightsquigarrow q_m$ for a run from $q_0$ to $q_m$.
We write $\Runs(q, q')$ for the set of runs from the location $q$ to $q'$. 
A run is said to be \emph{canonical} if delay and switch transitions alternate.

 \paragraph{\bf Priced timed automata.}
 A priced timed automaton (PTA) is a timed automaton $\mathcal{A}$ = ($L$, $X$, $E$,
 $I$, $\pi$, $\psi$) augmented with a price functions $\pi: Q_{\mathcal{A}}
 \times \mathbb{R}^+\mapsto \mathbb{R}^+$ and $\psi: E \mapsto \mathbb{Z}^+$ 
 which assign prices (costs) for waiting at locations and taking edges, where
 $Q_{\mathcal{A}}$ is the set of the configurations of timed automaton
 $\mathcal{A}$. 
   Let $\mathcal{A}$ be a PTA and $\rho = q_0'
   \xrightarrow{\tau_1} q_1 \xrightarrow{e_1} q_1' \xrightarrow{\tau_2} q_2
   \xrightarrow{e_2} q_2' \cdots \xrightarrow{\tau_m} q_m \xrightarrow{e_m} 
  q_m'$ be a canonical run of $\mathcal{T}_\mathcal{A}$. 
  Then the cost $C(\rho)$ of the run $\rho$ is equal to $C_d(\rho) + C_s(\rho)$ where
 $ C_d(\rho) = \sum_{k = 1}^{m} \pi(q_{k-1}', \tau_k) \text{ and } C_s(\rho) =
  \sum_{k=1}^{m} \psi(e_k)$ are the \emph{duration} and \emph{switching} costs
  of $\rho$ respectively.  
  
  A priced timed transition system $\mathcal{T}$ is a tuple $(L, X, E)$ where
  \begin{inparaenum}[(i)]
  \item $L$ is a set of locations
  \item $X$ is a finite set of clocks
  \item $E$ is a set of transitions.
  \end{inparaenum}
  A \emph{configuration} is a tuple $(\ell, \nu, u)$, where $\ell \in L$ is a location, $\nu$ is a clock valuation over set $X$ and 
 $u \in \mathbb{R}$ is current accumulated price. Let $Q_{\mathcal{T}}$ be the set of configurations for
 timed transition system $\mathcal{T}$. There are two types of transitions defined over $Q_{\mathcal{T}}$:
 \begin{itemize}
  \item Delay, $E^\tau \subseteq Q_{\mathcal{T}} \times \mathbb{R}^+ \times Q_{\mathcal{T}}$: $(\mathcal{\ell}, \nu, u) \xrightarrow{t} (\mathcal{\ell}, \nu+t, u')$, where $t \in \mathbb{R}^+$
  \item Switch, $E^e \subseteq Q_{\mathcal{T}} \times 2^X \times Q_{\mathcal{T}}$: $(\mathcal{\ell}, \nu, u) \xrightarrow{Y} (\mathcal{\ell}', \nu[Y{:=}0], u')$ where $Y \subseteq X$.
 \end{itemize}
We write $E = E^{\mathcal{\tau}} \cup E^{e}$ for the set of transitions of timed transition system $\mathcal{T}$.
A priced timed transition system is said to be \emph{canonical} if for its every run, delay transitions and switch transitions occur 
in the strict alternation. Observe that two consecutive delay transitions like $(\ell_1, \nu_1, u_1) \xrightarrow{t_1} (\ell_2, \nu_2, u_2) \xrightarrow{t_2} (\ell_3, \nu_3, u_3)$
cannot be combined together as $(\ell_1, \nu_1, u_1) \xrightarrow{t_1 + t_2} (\ell_3, \nu_3, u_3)$ because for non-linear price functions
such clubbing may not yield valid transitions. 
Similarly, there cannot be consecutive switch transitions without zero delay transition 
between them. 

Let $\mathcal{A}$ = ($L$, $X$, $E$, $I$, $\pi$, $\psi$) be a priced timed automaton. 
The semantics of $\mathcal{A}$ are given by a canonical priced timed transition system 
$\mathcal{T}_\mathcal{A} = (L_\mathcal{A}, X_\mathcal{A}, E_\mathcal{A})$ such that,
$L_{\mathcal{A}} = L$, $X_{\mathcal{A}} = X$, $E_{\mathcal{A}} =
E^{\tau}_{\mathcal{A}} \cup E^{e}_{\mathcal{A}}$, s.t.:
\begin{itemize}
\item $E^{\tau}_{\mathcal{A}} = \{(\mathcal{\ell}, \nu, u) \xrightarrow{t} (\mathcal{\ell}, \nu+t, u+\pi(\ell, t)) \; | \; t \in \mathbb{R}^+ \; \mbox{and} \; \forall \delta \in \mathbb{R}^+, 0 \leq \delta \leq t \Rightarrow (\nu+\delta)
  \vDash I(\mathcal{\ell})\}$
\item $E^{e}_{\mathcal{A}} = \{ (\mathcal{\ell}, \nu, u) \xrightarrow{Y} (\mathcal{\ell}', \nu[Y{:=}0], u+\psi(\gamma)) \; | \; 
  Y \subseteq X, \; \mbox{transition} \; \gamma = (\mathcal{\ell}, g, Y,
  \mathcal{\ell}') \in E,\; 
  \nu \vDash g, \; \nu \vDash I(\mathcal{\ell}) \; \mbox{and} \;  \nu[Y{:=}0]] \vDash I(\ell')\}$.
\end{itemize}
A run of the transition system $\mathcal{T}_\mathcal{A}$ starts with some configuration $(\ell, \nu, u_0)$ where $\ell \in L$,
  $\nu \in ({\mathbb{R}^+})^X$ and $u_0 \in \mathbb{R}$.
 We do not explicitly specify initial configuration in our definition of priced
 timed automaton. 
 
 \begin{definition}[Cost-optimal reachability problem] 
   Let $\mathcal{A}$ be a priced timed automaton. 
   Given two locations $\ell, \ell'$~of $A$, the optimal cost $\OptCost(\ell, \ell')$, of
   reaching $\ell'$ from $\ell$ is defined as  
   \[
   \OptCost(\ell, \ell') = \inf_{\rho \in \Runs(\ell, \ell')} C(\rho).
   \]
   Given priced timed automaton $\mathcal{A}$, locations $\ell$, $\ell'$, and a budget
   $B \in \Real$ the cost-optimal reachability problem is to decide whether 
   $\OptCost(\ell, \ell') \leq B$. 
 \end{definition}

\paragraph{\bf Summary of Our Results.}
 Our first result (Section~\ref{sec:undec}) is that the optimal cost
 reachability problem for general  priced timed automata is undecidable. 
 \begin{theorem}
   \label{thm1}
   Cost-optimal reachability problem for nonlinearly priced timed automata is
   undecidable.  
 \end{theorem}

 Given this negative result it is justifiable to look for various restricted
 subclasses of price functions in order to recover decidable variants. 
The first subclass that we consider is piece-wise linear price functions. 
A \emph{piecewise linear price function}  $f : \Rplus {\to} \Real$ can be
represented as tuple     
$(P^{\mathcal{\ell}}, Y_P^{\mathcal{\ell}}, Y_I^{\mathcal{\ell}})$ where
\begin{itemize}
\item $P^{\mathcal{\ell}} = \langle p_1=0, p_2, \ldots, p_n \rangle$ is an increasing sequence  of $n$-points in time.
  First point $p_1$ is always at time value zero. Thus, $p_1 (= 0) < p_2 < \ldots < p_n$ holds.
\item $Y_P^{\mathcal{\ell}} = \langle y_{p_1}, y_{p_2}, \ldots, y_{p_n} \rangle$ is a sequence of prices such that $y_{p_i} = f_\mathcal{\ell}(p_i)$.
\item $Y_I^{\mathcal{\ell}} = \langle (m_1, c_1), (m_2, c_2), \ldots, (m_n, c_n) \rangle$ is a sequence of $n$ tuples. 
  Time intervals formed by points in $P^{\mathcal{\ell}}$ are $I_1 \triangleq (p_1, p_2), \cdots, I_n \triangleq (p_n, +\infty)$.
  Again let $I = \langle I_1, \ldots, I_n \rangle$ be the sequence of intervals. Value of piecewise linear price function $f_\mathcal{\ell}$
  in the interval $I_k$ is given by parameters in tuple $(m_k, c_k)$, such that if $\tau \in I_k$, $f(\tau) = m_k \tau + c_k$.
\end{itemize}  
We call tuple $(P^{\mathcal{\ell}}, Y_P^{\mathcal{\ell}}, Y_I^{\mathcal{\ell}})$ as \emph{structure} of function $f$.
We call a timed automaton is \emph{piecewise linearly priced} if for all
configuration $(\ell, \nu)$ we have that 
$\pi((\ell, \nu), \tau) = f_\mathcal{\ell}(\tau)$,
where $f_\mathcal{\ell} = (P^{\mathcal{\ell}}, Y_P^{\mathcal{\ell}},
Y_I^{\mathcal{\ell}})$, is a piecewise linear function defined over interval
$[0, +\infty)$ 
  and all the constants appearing in its structure are integers.  
  Observe that the standard definition of linearly priced timed automata can be
  casted as a special case of piecewise linearly priced
  timed automata such that $f_\mathcal{\ell} =  (P^{\mathcal{\ell}},
  Y_P^{\mathcal{\ell}}, Y_I^{\mathcal{\ell}})$, where 
  $P^{\mathcal{\ell}} = \langle 0 \rangle$, $Y_P^{\mathcal{\ell}} = \langle 0
  \rangle$, and $Y_I^{\mathcal{\ell}} = \langle (k_\mathcal{\ell}, 0) \rangle$
  such that $k_\mathcal{\ell}$ is rate of change of price at location  
  $\mathcal{\ell}$.
  For LPTA the cost-optimal reachability problem is
  known to be PSPACE-complete~\cite{Bouyer:2007:ORP:1288667.1288679}. 
  In Section~\ref{sec:dec} we show the following key result of piecewise linearly priced
  timed automata.
  \begin{theorem}
    \label{thm:main2}
    The cost-optimal reachability problem for piecewise linearly priced timed
    automaton is PSPACE-complete. 
  \end{theorem}
  This result can easily be extended to piecewise-concave priced timed
  automata~\cite{DBLP:conf/formats/JurdzinskiT08}.
  
  We also study more general Lipschitz continuously priced timed automata. 
  We say that a function $f: \mathbb{R} \mapsto \mathbb{R}$ is Lipschitz
  continuous function, if there exists a constant $K \geq 0$, called Lipschitz
  constant of $f$, s.t. $\|f(x) - f(y) \| \leq K \|x - y \|$ for all $x, y$
  in the domain of $f$. 
  A timed automaton is then called \emph{Lipschitz continuous priced} if price
  functions $\pi((\ell, \nu), \tau) = f_\mathcal{\ell}(\tau)$, 
  are Lipschitz continuous for every location $\ell$ and there exists a constant
$T$ such that all the clock valuations are bounded from above by $T$.
 For this class of functions the optimal reachability problem may not be
computable due to optimal occurring at non-rational points. 
For this reason we study the following approximate optimal problem. 
\begin{definition}[$\varepsilon$-Cost-optimal reachability problem] 
  Let $\mathcal{A}$ be a priced timed automaton. Given $\varepsilon >0$ and
  two locations $\ell$, $\ell'$ of $\mathcal{A}$, a budget $B \in \Rplus$,  the
  $\varepsilon$-optimal cost problem 
  is to decide whether $\OptCost(\ell, \ell') \leq B + \varepsilon$.
 \end{definition}
We show in Section~\ref{sec:dec} the following result for Lipschitz-continuous priced timed automata.
\begin{theorem}
\label{thm3}
The $\varepsilon$-Cost-optimal reachability is decidable for Lipschitz-continuous
priced timed automata.
\end{theorem}
Finally, in Section~\ref{sec:tool} we give details of our implementation to
solve step-bounded cost-optimal reachability problem for general priced timed
automata.

\section{Undecidability}
\label{sec:undec}
This section is devoted to the proof of Theorem~\ref{thm1}.
We prove this result by reducing the halting problem for two 
counter machines to the cost-optimal reachability problem for priced timed automata.
 A \emph{two-counter machine} $M$ is a tuple $(L, C)$ where ${L = \set{\ell_0,
    \ell_1, \ldots, \ell_n}}$ is the set of instructions including a
distinguished terminal instruction $\ell_n$ called HALT, and the set 
${C = \set{c_1, c_2}}$ of two \emph{counters}.  
The instructions $L$ are of the type:
(1) (increment $c$) $\ell_i : c := c+1$;  goto  $\ell_k$,
(2) (decrement $c$) $\ell_i : c := c-1$;  goto  $\ell_k$,
(3) (zero-check $c$) $\ell_i$ : if $(c >0)$ then goto $\ell_k$
  else goto $\ell_m$,
where $c \in C$, $\ell_i, \ell_k, \ell_m \in L$.
A configuration of a two-counter machine is a tuple $(\ell, c, d)$ where
$\ell \in L$ is an instruction, and $c, d \in \Nat$ is the value of counters $c_1$
and $c_2$, resp. 
A run of a two-counter machine is a (finite or infinite) sequence of
configurations $\seq{k_0, k_1, \ldots}$ where $k_0 = (\ell_0, 0, 0)$ and the
relation between subsequent configurations is governed by transitions between
respective instructions. 
The \emph{halting problem} for a two-counter machine asks whether 
its unique run ends at the terminal instruction $\ell_n$.
The halting problem for two-counter machines is known~\cite{Min67} to be undecidable.

\begin{proof}[Proof of Theorem~\ref{thm1}]
  We reduce the reachability problem of two counter machines to an instance of the
  cost-optimal reachability problem $\OptCost(q, q')$ for priced timed automata
  $\mathcal{A}$ such that  
  desired configuration of two counter machine
  is reachable from its initial configuration 
  iff there is a run in the automaton $\mathcal{A}$ from $q$ to $q'$ of cost exactly zero.  
  
\begin{figure}[t]
 \centering
 
 \tikzstyle{location}=[rectangle, rounded corners, draw=blue!50,fill=blue!10,thick]
 \tikzstyle{transition}=[rectangle,draw=black!50,fill=black!20,thick]
 
 \begin{tikzpicture}[node distance=2.5cm,auto]  

  \node[draw,fill=black!10] (in_dec) {$\begin{array}{lll}
		    x&=&\encode(c)\\
		    y&=&\encode(d)\\
		    w&=&0\\
		    z&=&0
		  \end{array}$};
  
  \node[initial,initial text=,location] (l0) [right of=in_dec]  {$l_0$};
  \node[location] (l1) [right of=l0] {$l_1$};  
  \node[location] (l2) [right of=l1] {$l_2$};  

  \draw[->] (l0) to  node {$z > 0?$} node [swap] {$x,w:=0$} (l1);
  \draw[->] (l1) to node {$z=1?$} node [swap] {$\begin{array}{l}
                                           x:=0,\\
                                           z:=0
                                          \end{array}$} (l2);

  \draw[->] (l1) edge [loop above] node {$\begin{array}{l}
                                                                                    y=1?\, y:=0
                                          \end{array}$} ();

                                          \draw[->] (l0) edge [loop above] node {$\begin{array}{l}
                                                                             y=1?\, y:=0
                                          \end{array}$} ();

  \node [node distance=0.5cm,below of=l2] {$\{z=0\}$};
  
  \node[draw,fill=black!10] (out_dec) [right of=l2] {$\begin{array}{lll}
    x&=&0\\
    y&=&\encode(d)\\
    w&=&\encode(c-1)\\
    z&=&0
  \end{array}$};
  
  \node [node distance=1.3cm,below of=l0] {$\pi_{l_0} = (1-x-\frac{t}{2})^2$};
  \node [node distance=1.3cm,below of=l1] {$\pi_{l_1} = 0$};
  \node [node distance=1.3cm,below of=l2] {$\pi_{l_2} = 0$};
  
 \end{tikzpicture}
 \caption{Decrement $c$ module}
 \label{fig:undec-decr}
 \end{figure}
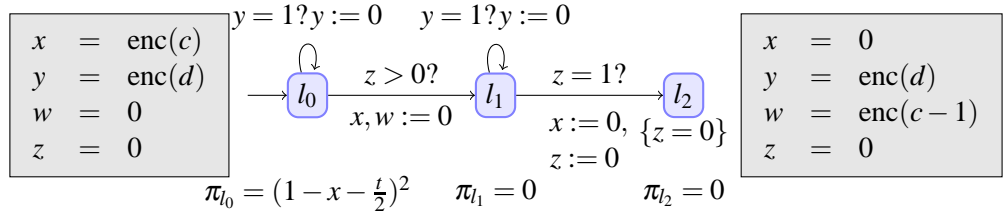 
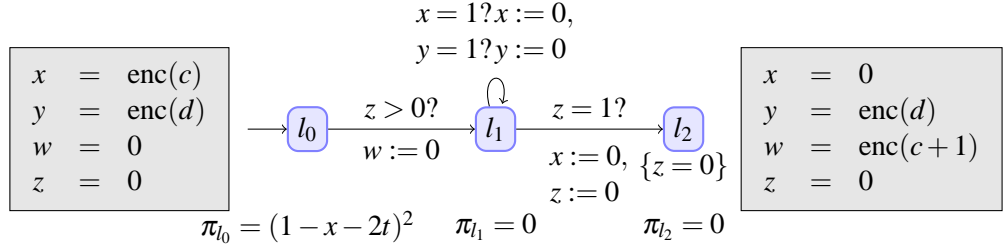
\begin{figure}
 \centering 
 \tikzstyle{location}=[rectangle, rounded corners, draw=blue!50,fill=blue!10,thick]
 \tikzstyle{transition}=[rectangle,draw=black!50,fill=black!20,thick]

 \begin{tikzpicture}[node distance=2.5cm,auto]  

  \node[draw,fill=black!10] (module_in) {$\begin{array}{lll}
		    x&=&\encode(c)\\
		    y&=&\encode(d)\\
		    w&=&0\\
		    z&=&0
		  \end{array}$};
  
  \node[initial,initial text=,location] (l0) [right of=module_in]  {$l_0$};
  \node[location] (l1) [right of=l0] {$l_1$};  
  \node[location] (l2) [right of=l1] {$l_2$};  

  \draw[->] (l0) to  node {$z > 0?$} node [swap] {$w:=0$} (l1);
  \draw[->] (l1) to node {$z=1?$} node [swap] {$\begin{array}{l}
                                           x:=0,\\
                                           z:=0
                                          \end{array}$} (l2);

  \draw[->] (l1) edge [loop above] node {$\begin{array}{l}
                                           x=1?\, x:=0,\\
                                           y=1?\, y:=0
                                          \end{array}$} ();
  
  \node [node distance=0.5cm,below of=l2] {$\{z=0\}$};
  
  \node[draw,fill=black!10] (module_out) [right of=l2] {$\begin{array}{lll}
    x&=&0\\
    y&=&\encode(d)\\
    w&=&\encode(c+1)\\
    z&=&0
  \end{array}$};
  
  \node [node distance=1.3cm,below of=l0] {$\pi_{l_0} = (1-x-2t)^2$};
  \node [node distance=1.3cm,below of=l1] {$\pi_{l_1} = 0$};
  \node [node distance=1.3cm,below of=l2] {$\pi_{l_2} = 0$};
  
 \end{tikzpicture}
 \caption{Increment $c$ module}
 \label{fig:undec-incr}
 \end{figure}

  Let $\mathcal{M}$ be the instance of the two counter machine having counters
  $c$ and $d$. 
  We construct a PTA $\mathcal{A}$ from $\mathcal{M}$ using suitable encoding.
  Valid runs of $\mathcal{M}$ are mapped to valid runs of $\mathcal{A}$ such that
  their cost is exactly zero. Figure~\ref{fig:undec-decr}
  and \ref{fig:undec-incr} 
  describes the module 
  simulating counter decrement  instruction of
  $\mathcal{M}$. 
  PTA  $\mathcal{A}$ is constructed by composing the various modules. 
  $\mathcal{A}$ uses four clocks -- $x$, $y$, $w$ and $z$, out of which $x$ and
 $y$ encode counters $c$ and $d$ as $x=1-\frac{1}{2^c}$ and
 $y=1-\frac{1}{2^d}$. Let $\encode(\cdot)$ denote this encoding function. 
 Testing whether $c$ is zero amounts to
 testing $x$ is zero in the guards of $\mathcal{A}$. Figure~\ref{fig:undec-decr}
 describes decrement operation on counter $c$. It shows clock valuations before
 entering the module and after 
 exiting the module when simulation is correct. Let $t$ be the amount of time
 spent in location $l_0$ during simulation. Let $(x,y,z,w)=(1-\frac{1}{2^c}, 1-\frac{1}{2^d}, 0, 0)$ be the configuration 
 on entering $l_0$.  
 We want to ensure that the time spent at $l_0$ is $t=\frac{1}{2^{c-1}}$. 
 The self loop at $l_0$ ensures that the value of $y$ never crosses 1. 
  If so, the new values of $x,y,z,w$ respectively are 
 $0, 1-(\frac{1}{2^d}-\frac{1}{2^{c-1}})$ or $\frac{1}{2^{c-1}}-\frac{1}{2^d}$, $\frac{1}{2^{c-1}}, 0$.
Note that the new value of $y$ after elapse of time $t$ is 
 $1-(\frac{1}{2^d}-\frac{1}{2^{c-1}})$ or $\frac{1}{2^{c-1}}-\frac{1}{2^d}$ depending on whether 
 $d > c$ or not. 
   A time of $1-\frac{1}{2^{c-1}}$ is spent 
 at location $l_1$. This gives us the configuration 
 $0, 1-\frac{1}{2^d}, 0, 1-\frac{1}{2^{c-1}}$ on reaching $l_2$. Note that 
 the self loop on $y$ at location $l_1$ 
 helps in regaining the value of $y$ to be $1-\frac{1}{2^d}$ 
 in the case when $d>c$.  
 Note that the cost is 0 iff $t=\frac{1}{2^{c-1}}$. 
  Thus, only correct simulation 
 incurs zero price. Likewise increment module in figure~\ref{fig:undec-incr}
 correctly works when $t=\frac{1-x}{2}$.  
 
 Observe that after every increment or decrement operation, the value of clock
 $x$ moves to clock $w$. 
 Hence, in order to composing $\mathcal{A}$ from individual modules we need to
 swap the roles of clocks $x$ and $w$ in every alternate modules. 
 Let $\langle c_1,d_1 \rangle$ be initial configuration and $\langle
 c_2,d_2 \rangle$ be target configuration of $\mathcal{M}$. 
 They map to clock valuation $\nu_1=(\encode(c_1), \encode(d_1), 0,0)$ 
 and $\nu_2=(\encode(c_2), \encode(d_2), 0,0)$ respectively. 
 To make $\nu_1$ and $\nu_2$ separate locations, we can scale all
 constants $\nu_1, \nu_2$ and $\mathcal{A}$ so as to make clock values in
 $\nu_1$ and $\nu_2$ integers. 
The construction is now complete. 
\qed
 \end{proof}

\section{Decidable Subclasses}
\label{sec:dec}
\paragraph{\bf Priced Timed Bisimulations.}
Let $\mathcal{A}$ and $\mathcal{B}$ be timed automata with their timed
transition systems $\mathcal{S}_\mathcal{A}$ and $\mathcal{S}_\mathcal{B}$. 
Let $\mathcal{Q}_\mathcal{A}$ and $\mathcal{Q}_\mathcal{B}$ respective sets of
configurations. 
A binary symmetric relation $\mathcal{R}$ over $\mathcal{Q}_\mathcal{A} \times
\mathcal{Q}_\mathcal{B}$ is a \emph{strong timed bisimulation relation} 
iff for all $a \in (\mathbb{R}^+ \cup 2^X)$
\begin{itemize}
 \item 
   if $q_1 \xrightarrow{a} q_1' $ and $q_1 \mathcal{R} q_2$ then there exists
   transition  $q_2 \xrightarrow{a} q_2'$ such that $q_1' \mathcal{R} q_2'$  
 \item 
   conversely, if $q_2 \xrightarrow{a} q_2'$ and $q_1 \mathcal{R} q_2$ then
   there exists transition 
   $q_1 \xrightarrow{a} q_1'$ such that $q_1' \mathcal{R} q_2'$,
\end{itemize}
where $q_1,q_1' \in \mathcal{Q}_\mathcal{A}$ and $q_2,q_2' \in \mathcal{Q}_\mathcal{B}$.
The relation $\mathcal{R}$ is \emph{strong timed bisimilarity} or \emph{strong
  timed bisimulation equivalence} if it is the largest strong  timed
bisimulation relation such that $\mathcal{R} \subseteq \mathcal{Q}_\mathcal{A}
\times \mathcal{Q}_\mathcal{B}$. 
Timed automata $\mathcal{A}$ and $\mathcal{B}$ are \emph{strong timed bisimilar}
if there exists such $\mathcal{R}$.  

  Let $\mathcal{A}$ and $\mathcal{B}$ be priced timed automata with their priced timed transition systems $\mathcal{T}_\mathcal{A}$ 
  and $\mathcal{T}_\mathcal{B}$. Let $\mathcal{P}_\mathcal{A}$ and $\mathcal{P}_\mathcal{B}$ be respective sets of priced configurations.
  A strong timed bisimilarity $\sim$ is said to \emph{price preserving} if for every $a \in (\mathbb{R}^+ \cup 2^X)$ 
  \begin{itemize}
  \item if $(q_1, u_1) \xrightarrow{a} (q_1', u_1') $ is in $\mathcal{T}_\mathcal{A}$ and $q_1 \sim q_2$ then there exists transition
    $(q_2, u_2) \xrightarrow{a} (q_2', u_2') $ in $\mathcal{T}_\mathcal{B}$ such that $q_1' \sim q_2'$ and $(u_1' - u_1) = (u_2' - u_2)$
  \item conversely, if $(q_2, u_2) \xrightarrow{a} (q_2', u_2') $ is in $\mathcal{T}_\mathcal{B}$ and $q_1 \sim q_2$ then there exists transition
    $(q_1, u_1) \xrightarrow{a} (q_1', u_1') $ in $\mathcal{T}_\mathcal{A}$ such that $q_1' \sim q_2'$ and $(u_1' - u_1) = (u_2' - u_2)$
  \end{itemize}
  where $(q_1, u_1), (q_1', u_1') \in \mathcal{P}_\mathcal{A}$ and 
  $(q_2, u_2), (q_2', u_2') \in \mathcal{P}_\mathcal{B}$.

\begin{lemma} \label{lem:bisim}
 If $\mathcal{A}$ and $\mathcal{B}$ are two priced timed automata with price preserving timed bisimilarity $\sim$, then
 for any $k$ length run $\varrho_\mathcal{A}^{(k)}$ in $\mathcal{A}$, where
 $  
  \varrho_\mathcal{A}^{(k)} = (q_\mathcal{A}^0, u_0)  \xrightarrow{a_1}  (q_\mathcal{A}^1, u_1)  \xrightarrow{a_2} 
  (q_\mathcal{A}^2, u_2)  \xrightarrow{a_3}  \cdots  \xrightarrow{a_{k-1}}  (q_\mathcal{A}^{k-1}, u_{k-1})  \xrightarrow{a_k} 
  (q_\mathcal{A}^k, u_k),
 $
 there is a run $k$ length run $\varrho_\mathcal{B}^{(k)}$ in $\mathcal{B}$, where
 $
  \varrho_\mathcal{B}^{(k)} = (q_\mathcal{B}^0, u_0)  \xrightarrow{a_1}  (q_\mathcal{B}^1, u_1)  \xrightarrow{a_2}  
  (q_\mathcal{B}^2, u_2)  \xrightarrow{a_3}  \cdots  \xrightarrow{a_{k-1}}  (q_\mathcal{B}^{k-1}, u_{k-1})  \xrightarrow{a_k} 
  (q_\mathcal{B}^k, u_k),
 $
 such that, for every $0 \leq i \leq k$, $q_\mathcal{A}^{i} \sim q_\mathcal{B}^{i}$ holds. $u_0$ is initial credit. 
\end{lemma}

As the choice of initial credit is arbitrary and the cost of a run does not depend on the value of initial credit, we claim following lemma.
\begin{lemma} \label{lem:bisim2}
 Let $\mathcal{A}$ and $\mathcal{B}$ be two priced timed automata with price preserving timed bisimilarity $\sim$. Then following
 statements are true.
 \begin{enumerate}
  \item \label{itm:stmt1} for every run $\rho_\mathcal{A}$ in $\mathcal{A}$, there exists a run $\rho_\mathcal{B}$ in $\mathcal{B}$ s. t. cost 
  $C(\rho_\mathcal{A}) = C(\rho_\mathcal{B})$
  \item for every run $\rho_\mathcal{B}$ in $\mathcal{B}$, there exists a run $\rho_\mathcal{A}$ in $\mathcal{A}$ s. t. cost 
  $C(\rho_\mathcal{A}) = C(\rho_\mathcal{B})$
 \end{enumerate}
\end{lemma}

\subsection{Proof of Theorem~\ref{thm:main2}} 
\label{sec:pwl}                
 \begin{lemma}
  For every piecewise linearly priced timed automaton (PwLPTA), there exists linearly priced timed automaton with price preserving strong bisimulation between them.
 \end{lemma}

 \begin{proof}
  We prove this lemma by constructing LPTA explicitly from a given PwLPTA. Rest of this section explains construction and 
  lemma~\ref{lem:pwl} proves its correctness. 
  
 \end{proof}

 \paragraph*{Construction of LPTA} Let $\mathcal{A}$ = ($L_\mathcal{A}$, $X_\mathcal{A}$, $E_\mathcal{A}$, $I_\mathcal{A}$, $\pi_\mathcal{A}$, 
 $\psi_\mathcal{A}$) be a PwLPTA. 
 We construct LPTA $\mathcal{B}$ = ($L_\mathcal{B}$, $X_\mathcal{B}$, $E_\mathcal{B}$, $I_\mathcal{B}$, $\pi_\mathcal{B}$, 
 $\psi_\mathcal{B}$) from PwLPTA $\mathcal{A}$ as follows:
 \begin{itemize}  
  \item Let $\mathcal{\ell} \in L_\mathcal{A}$ be some location of $\mathcal{A}$. Delay price function for location $\mathcal{\ell}$,
  $\pi_\mathcal{A}(\mathcal{\ell}, \tau) = f_\mathcal{\ell}(\tau)$, is piecewise linear with respect to $\tau$. $f_\mathcal{\ell}$ is
  given by integer restricted structure $(P^\mathcal{\ell}, Y^\mathcal{\ell}_P, Y^\mathcal{\ell}_I)$, where
  \begin{itemize}
   \item $P^{\mathcal{\ell}} = \langle p_1=0, p_2, \ldots, p_n \rangle$   
   \item $Y_P^{\mathcal{\ell}} = \langle y_{p_1}, y_{p_2}, \ldots, y_{p_n} \rangle$ 
   \item $Y_I^{\mathcal{\ell}} = \langle (m_1, c_1), (m_2, c_2), \ldots, (m_n,
     c_n) \rangle$ with the following interval sequence 
     \[
     I = \langle I_1 \triangleq (p_1, p_2), I_2 \triangleq (p_2, p_3), \cdots,
     I_n \triangleq (p_n, p_{n+1}=+\infty ) \rangle.
     \]
  \end{itemize}
  We associate each $p_i \in P^\mathcal{\ell}$ and each $I_j \in I$ with locations of $L_\mathcal{B}$. This association is captured by mapping 
  $\alpha^\mathcal{\ell}$ such that $\alpha^\mathcal{\ell}(p_i) = \mathcal{\ell}^{p_i}$ and $\alpha^\mathcal{\ell}(I_j) = 
  \mathcal{\ell}^{(p_j, \, p_{j+1})}$. Here, $\mathcal{\ell}^{p_i}$ and $\mathcal{\ell}^{(p_j, \, p_{j+1})}$
  are the names of locations of $\mathcal{B}$.
  We define another mapping $\beta^\mathcal{\ell}(I_j)$ which returns $j$\textsuperscript{th} entry in the sequence $Y_I^\mathcal{\ell}$. This mapping is
  useful for retrieving parameters of delay cost function in the interval $I_j$.
  Let $\theta^\mathcal{\ell} = \cup_{i=1}^{n}\{\mathcal{\ell}^{p_i}, \mathcal{\ell}^{(p_i, \, p_{i+1})}\}$.
  $\theta^\mathcal{\ell}$ denotes locations in $L_\mathcal{B}$ generated from location $\mathcal{\ell} \in L_\mathcal{A}$.
  Then $L_\mathcal{B} := \cup_{\mathcal{\ell} \in L_\mathcal{A}} \theta^\mathcal{\ell}$. 
  \item We add one extra clock named $x$ to $\mathcal{B}$. Thus, $X_\mathcal{B} := X_\mathcal{A} \cup \{x\}$. This clock measures
  time spent at every location of $\mathcal{A}$. Whenever a run enters any location of $\mathcal{A}$, $x$ is reset to zero.
  
  \item An edge $e = (l, \varphi, \lambda, l') \in E_\mathcal{B}$ iff 
  there is an edge $e' = (\mathcal{\ell}, \chi, \xi, \mathcal{\ell}') \in E_\mathcal{A}$ such that 
  \begin{itemize}      
   \item either $\alpha^\mathcal{\ell}(p_i) = l$ or $\alpha^\mathcal{\ell}(I_i) = l$
   \item either $\alpha^\mathcal{\ell}(p_j) = l'$ or $\alpha^\mathcal{\ell}(I_j) = l'$   
   \item $\varphi := \left\{ 
     \begin{array}{ll}
      \chi \wedge (x = p_i) & \mbox{if } \alpha^\mathcal{\ell}(p_i) = l\\
      \chi \wedge (x \in I_i) & \mbox{if } \alpha^\mathcal{\ell}(I_i) = l      
     \end{array} \right.$
   \item $\lambda := \xi \cup \{x\}$
  \end{itemize}
  \item Location invariant, $I_\mathcal{B}(l) = I_\mathcal{A}(\ell)$ iff $l \in \theta^\ell$
  \item Location price, $\pi_\mathcal{B}(l) := \left\{ 
     \begin{array}{ll}
       0 & \mbox{if } \alpha^\mathcal{\ell}(p_i) = l\\
      m_i & \mbox{if } \alpha^\mathcal{\ell}(I_i) = l \mbox{ and } \beta^\mathcal{\ell}(I_i) = (m_i, c_i)
     \end{array} \right.$
  \item Edge price, $\psi_\mathcal{B}(e) := \left\{ 
     \begin{array}{ll}
      \psi_\mathcal{A}(e') + y_{p_i} & \mbox{if } \alpha^\mathcal{\ell'}(p_i) = l'\\
      \psi_\mathcal{A}(e') + c_i & \mbox{if } \alpha^\mathcal{\ell'}(I_i) = l' \mbox{ and } \beta^\mathcal{\ell'}(I_i) = (m_i, c_i)
     \end{array} \right.$  
  
 \end{itemize}

 Let $\ell$ and $m$ be locations of $\mathcal{A}$ and $\mathcal{B}$ respectively. We define following relation between $\ell$ and $m$,
 $\Upsilon = \{(\ell, m) \, \vert \, m \in \theta^\ell \}$.
 \begin{lemma}\label{lem:pwl}
  $\Upsilon$ is price preserving timed bisimilarity.  
 \end{lemma}
 \begin{proof}
   Let $\mathcal{A}$ = ($L_\mathcal{A}$, $X_\mathcal{A}$, $E_\mathcal{A}$, $I_\mathcal{A}$, $\pi_\mathcal{A}$, 
 $\psi_\mathcal{A}$) be a PwLPTA and $\mathcal{B}$ = ($L_\mathcal{B}$, $X_\mathcal{B}$, $E_\mathcal{B}$, $I_\mathcal{B}$, $\pi_\mathcal{B}$, 
 $\psi_\mathcal{B}$) be LPTA constructed from $\mathcal{A}$ using above construction.
 
 If part: Let $t$ be delay and $\lambda$ be set of clocks to be reset in $\mathcal{A}$. Now consider following transition in 
 $\mathcal{T}_\mathcal{A}$, $(l_1, \nu_1, u_1) \xrightarrow{(t, \lambda)} (l'_1, \nu'_1, u'_1)$. Now we try to find simulating
 transition in $\mathcal{B}$ under relation $\Upsilon$. We claim its $(l_2, \nu_1:0, u_1) \xrightarrow{(t, \lambda)} (l'_2, \nu'_1:0, u'_1)$.
 To hold this claim, we choose $l_2 \in L_\mathcal{B}$ such that delay $t$ matches with expected interval of $l_2$. If $t=p_i$ for some 
 $i$ then $l_2 = \alpha^{l_1}(p_i)$. Otherwise $t$ will match with some interval $I_j$. So $l_2 = \alpha^{l_1}(I_j)$. Thus, $(l_1,
 l_2) \in \Upsilon$ holds. To place edge in $\mathcal{B}$, construction mandates $(l'_1, l'_2) \in \Upsilon$. Also the clocks in $X_\mathcal{A}$ change identically.
 Now, let's verify that prices are preserved. For the case where $t=p_i$, $(u'_1 - u_1) = y_{p_i} + \psi_\mathcal{A}((l_1, l_2))$.
 Verify that from construction yields same price difference. For the case where $t=I_j$, location price matters. Verify that
 rates at $l'_1$ and $l'_2$ are the same in the construction. Price change $(u'_1 - u_1) = m_j \cdot t + c_j + \psi_\mathcal{A}((l_1, l_2))$.
 Price offset $c_j$ is added to edge cost in the construction. Thus prices are preserved.
 
 Else if part: We consider following transition in 
 $\mathcal{T}_\mathcal{B}$, $(l_2, \nu_2:0, u_2) \xrightarrow{(t, \lambda)} (l'_2, \nu'_2:0, u'_2)$. We simulate it on $\mathcal{A}$
 to get $(l_1, \nu_2, u_1) \xrightarrow{(t, \lambda)} (l'_1, \nu'_2, u'_1)$. If $(l_1, l_2) \in \Upsilon$, then construction offers no choice
 but to choose $l'_1$ such that $(l'_1, l'_2) \in \Upsilon$ holds. $\nu'_2 := (\nu_2+t)[\lambda := 0]$ follows from construction. Verify that
 prices are preserved using the same argument as in if part of the proof.
 \end{proof}

Now we are in position to sketch the proof of Theorem~\ref{thm:main2}. 

 \begin{proof}[Proof of Theorem~\ref{thm:main2}]
  PSPACE-hardness follows from the fact that LPTA are nothing but PwLPTA with single piece and their cost-optimal reachability is 
  PSPACE-complete. We now explain a PSPACE algorithm for solving cost-optimal reachability for PwLPTA.
  We construct LPTA $\mathcal{B}$ for given piecewise linearly priced timed automaton $\mathcal{A}$ and 
  solve cost-optimal reachability on $\mathcal{B}$. Construction yields priced timed bisimilarity $\Upsilon$. 
  Using lemma~\ref{lem:bisim2},
  we get  $\mbox{OptCost}(l, l')$ $= \mbox{opt }\{ \mbox{OptCost}(m, m')\ \vert $ $\ (l,m) \in \Upsilon \mbox{ and } (l',m') \in \Upsilon \}
   $
  where $l$ and $l'$ are locations of $\mathcal{A}$, $m$ and $m'$ are locations of $\mathcal{B}$ and
  opt is either supremum or infimum. \qed
 \end{proof}

\subsection{Proof of Theorem~\ref{thm3}}
Before we sketch a proof of Theorem~\ref{thm3}, we introduce the concept of iterative approximation for nonlinear price functions.

 Let $\mathcal{A}$ = ($L$, $X$, $E$, $I$, $\pi$, $\psi$) be a priced timed
 automaton. If for some location $\mathcal{\ell}$, price function
 $\pi(\mathcal{\ell}, \tau)$ is nonlinear with respect to $\tau$, then 
 $\mathcal{A}$ is nonlinearly priced timed automaton (NLPTA).
 
 \begin{definition}
  We define a PwLPTA $\mathcal{A}_u$ = ($L$, $X$, $E$, $I$, $\pi_u$, $\psi$) be upper bound price approximation of $\mathcal{A}$, if for every location $\mathcal{
  \ell}$ and time $\tau$, $\pi_u(\mathcal{\ell}, \tau) \geq \pi(\mathcal{\ell}, \tau)$ and $\pi_u(\mathcal{\ell}, \tau)$ is piecewise linear
  in $\tau$ for a fixed $\mathcal{\ell}$.
  
  Similarly, a PwLPTA $\mathcal{A}_l$ = ($L$, $X$, $E$, $I$, $\pi_l$, $\psi$) is lower bound price approximation of $\mathcal{A}$, if for every location $\mathcal{
  \ell}$ and time $\tau$, $\pi_l(\mathcal{\ell}, \tau) \leq \pi(\mathcal{\ell}, \tau)$ and $\pi_l(\mathcal{\ell}, \tau)$ is piecewise linear
  in $\tau$ for a fixed $\mathcal{\ell}$.
 \end{definition}
 
 \begin{lemma}
  OptCost$_{\mathcal{A}_l}$($\ell$, $\ell'$) $\leq$ OptCost$_\mathcal{A}$($\ell$, $\ell'$) $\leq$ OptCost$_{\mathcal{A}_u}$($\ell$, $\ell'$)
 \end{lemma}

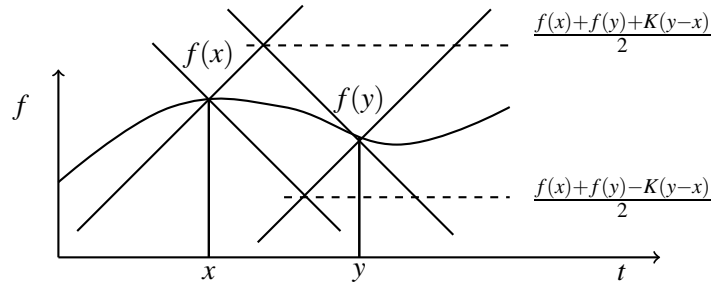
\begin{figure}[b]
\begin{center}
\begin{tikzpicture}[thick,outer sep=0,inner sep=0,minimum size=0]
\node (v1) at (-2.5,4.5) {};
\node (v2) at (-2.5,7) {};
\node (v3) at (5.5,4.5) {};
\draw[->]  (v1) edge (v2);
\draw[->]  (v1) edge (v3);
\draw  plot[smooth, tension=.7] coordinates {(-2.5,5.5) (-1,6.5) (0.5,6.5) (2,6) (3.5,6.5)};
\node (v8) at (-0.5,4.5) {};
\node (v10) at (1.5,4.5) {};
\node (v9) at (-0.5,6.6) {};
\node (v11) at (1.5,6.1) {};

\begin{scope}[xshift=1.25cm, yshift=-4.15cm]
\node (v4) at (-2.5,11.5) {};
\node (v7) at (-3.5,9) {};
\node (v6) at (-0.50,12) {};
\node (v5) at (0,9) {};
\draw  (v4) edge (v5);
\draw  (v6) edge (v7);
\draw  (v8) edge (v9);
\draw  (v10) edge (v11);
\end{scope}

\begin{scope}[xshift=3.25cm, yshift=-4.7cm]
\node (v4) at (-3.5,12.5) {};
\node (v7) at (-3.1,9.4) {};
\node (v6) at (0,12.5) {};
\node (v5) at (-0.5,9.5) {};
\draw  (v4) edge (v5);
\draw  (v6) edge (v7);
\draw  (v8) edge (v9);
\draw  (v10) edge (v11);
\end{scope}

\node [outer sep=0,inner sep=0,minimum size=0] at (-0.5,4.3) {$x$};
\node [outer sep=0,inner sep=0,minimum size=0] at (1.5,4.3) {$y$};
\node [outer sep=0,inner sep=0,minimum size=0] at (-0.5,7.2) {$f(x)$};
\node [outer sep=0,inner sep=0,minimum size=0] at (1.5,6.6) {$f(y)$};

\node [outer sep=0,inner sep=0,minimum size=0] (v12) at (0,7.32) {};
\node [outer sep=0,inner sep=0,minimum size=0] (v13) at (3.5,7.32) {};
\draw[style=dashed]  (v12) edge (v13);

\node [outer sep=0,inner sep=0,minimum size=0] (v14) at (0.5,5.3) {};
\node [outer sep=0,inner sep=0,minimum size=0] (v15) at (3.5,5.3) {};
\draw[style=dashed]  (v14) edge (v15);

\node [outer sep=0,inner sep=0,minimum size=0] (v16) at (5,7.5) {$\frac{f(x)+f(y)+K(y-x)}{2}$};
\node [outer sep=0,inner sep=0,minimum size=0] (v17) at (5,5.3) {$\frac{f(x)+f(y)-K(y-x)}{2}$};

\node [outer sep=0,inner sep=0,minimum size=0] at (-3.0,6.5) {$f$};
\node [outer sep=0,inner sep=0,minimum size=0] at (5,4.3) {$t$};
\end{tikzpicture}
\end{center}
\vspace{-0.5cm}
\caption{Upper and lower bounds for Lipschitz continuous function in the range $[x,y]$}
\label{fig:lip}
\end{figure}

Now we are in position to sketch the proof of Theorem~\ref{thm3}.
\begin{proof}[Proof of Theorem~\ref{thm3}]
 Let $f:\Real \mapsto \Real$ be Lipschitz continuous function with Lipschitz constant $K$. Let $x, y \in \Real$ be any two arbitrary points in 
 the interval $[x, y]$. The value of $f$ in $[x, y]$ is upper bounded by $\frac{f(x)+f(y)+K(y-x)}{2}$ and lower bounded by $\frac{f(x)+f(y)-K(y-x)}{2}$.
 Figure~\ref{fig:lip} shows calculation of these bounds for a Lipschitz continuous function. More precisely, for every $t \in [x, y]$,
 \[
 f(t) \in \Big[ \frac{f(x)+f(y)-K(y-x)}{2}, \, \frac{f(x)+f(y)+K(y-x)}{2} \Big].
 \]
 Assume that $f$ is a rational function. 
 We will first prove decidability of $\varepsilon$-optimal cost reachability problem using this assumption.
 Later we will drop this assumption.
 
 We now construct two piecewise linear price functions $f_l$ and $f_u$ such that $f_l(t) \leq f(t) \leq f_u(t)$ holds for $0 \leq t \leq T$. 
 Let $T \in \Rplus$ is a constant such that all clock valuations are bounded above by $T$.
 
 Let $\delta \in \mathbb{Q}^{+}, 0 < \delta \leq T$ be the sampling period. Choice for the value of $\delta$ is explained at the end of the proof.
 We sample $f$ at periodic intervals of $\delta$ in the interval $0 \leq t \leq T$. 
 We define a piecewise linear functions 
\[
 \begin{array}{lll}
   f_l(t) & = f(t) & \mbox{if} \ t = N \cdot \delta, \mbox{where} \ N \in \mathbb{N}\\
     & = \frac{f(N \cdot \delta) + f((N+1) \cdot \delta) - K\delta}{2} & \mbox{if} \ t \in \big( N\cdot \delta, (N+1) \cdot \delta \big), \mbox{where} \ N \in \mathbb{N}\\
     &&\\
   f_u(t) & = f(t) & \mbox{if} \ t = N \cdot \delta, \mbox{where} \, N \in \mathbb{N}\\
     & = \frac{f(N \cdot \delta) + f((N+1) \cdot \delta) + K\delta}{2} & \mbox{if} \ t \in (N\cdot \delta, (N+1) \cdot \delta), \mbox{where} \ N \in \mathbb{N}\\
  \end{array}
\]
 Let $\mathcal{A}$ be priced timed automaton with Lipschitz continuous price functions at all locations.
 We construct automata $\mathcal{A}_l$ and $\mathcal{A}_u$ by replacing price function at every location while keeping everything else unchanged. Specifically, if price function at location $\ell$ in $\mathcal{A}$ is $\pi^{(\ell)} = f$, then
 in $\mathcal{A}_l$, price at location $\ell$ is $\pi^{(\ell)}_l = f_l$. Likewise we assign price $\pi^{(\ell)}_u = f_u$ to location $\ell$ of $\mathcal{A}_u$.
 Observe that $\mathcal{A}_l$ and $\mathcal{A}_u$ are replicas of $\mathcal{A}$ except the difference in location price functions.
 Since, $\pi^{(\ell)}_l(t) \leq \pi^{(\ell)}(t) \leq \pi^{(\ell)}_u(t)$ holds for all locations $\ell$, 
 OptCost$_{\mathcal{A}_l}$($\ell$, $\ell'$) $\leq$ OptCost$_\mathcal{A}$($\ell$, $\ell'$) $\leq$ OptCost$_{\mathcal{A}_u}$($\ell$, $\ell'$) follows. 
 Now, for any single delay transition, $\sup \{\| \pi^{(\ell)}_u(t) - \pi^{(\ell)}_l(t) \| \} \leq \| K \delta \|$ over all $0 \leq t \leq T$.
 Let $D$ be the diameter of region graph, then $\sup \{\| \OptCost_{\mathcal{A}_l}(\ell, \ell') - \OptCost_{\mathcal{A}_u}(\ell, \ell') \| \} = \varepsilon \leq \| D K \delta \|$.
 This gives us the bound on $\varepsilon$. We choose $\delta = \frac{\| DK \|}{\varepsilon}$.
 
 In the above construction $f$ is evaluated only at sampling points. We can safely drop the rationality restriction of $f$ by approximating it by rational function
 $f'$ such that $\| f - f' \| \leq \frac{\| D K \delta \|}{2}$.
\qed
\end{proof}

\section{Step-Bounded Cost-Optimal Reachability Problem}
\label{sec:tool}
In this section we look into the following step-bounded cost-optimal
reachability problem for priced timed automata.
 \begin{definition}[Step-Bounded Cost-optimal reachability problem] 
   Let $\mathcal{A}$ be a priced timed automaton. 
   Given two locations $\ell, \ell'$ of $A$, step bound $N \in \Nat$, the step-bounded
   optimal cost $\OptCost_N(\ell, \ell')$, is defined as
   \[
   \OptCost_N(\ell, \ell') = \inf_{\rho \in \Runs_N(\ell, \ell')} C(\rho),
   \]
   where $\Runs_N(\ell, \ell')$ are the set of canonical runs between $\ell$ and $\ell'$ of
   length less than or equal to $N$. 
   Given priced timed automaton $\mathcal{A}$, locations $\ell$, $\ell'$, and a budget
   $B \in \Rplus$ the step-bounded cost-optimal reachability problem is to
   decide whether  $\OptCost_N(\ell, \ell') \leq B$. 
 \end{definition}
 In this section we extend the encoding of 
Audemard, Cimatti, Kornilowicz, and Sebastiani~\cite{bmc-ta} 
to solve step-bounded optimal-cost reachability
 problem for priced timed automata.
 After generating the encoding, we can feed it to SMT solver that support the
 theory corresponding to the price functions to solve the step-bounded
 cost-optimal  reachability problem. 
 
 \subsection{Audemard-Cimatti-Kornilowicz-Sebastiani Encoding for PTA}
\label{sec:encod}
Let $\mathcal{A}_1, \mathcal{A}_2, \ldots \mathcal{A}_n$ be the priced timed automata
which are composed into network of automata $\mathcal{A}$. These automata communicate using channels. Let $\eta$ be the set of 
channels used in $\mathcal{A}$. If $c$ is a channel, then $c!$ is send operation on the channel $c$ and $c?$ is the blocking 
receive operation on the channel $c$.

\noindent{\bf Original Encoding for Timed Automata.}
We generate SMT formula for each automata using encoding from Audemard
et. al. \cite{bmc-ta}. 
As per their scheme, 
we create one real variable for every clock and create separate one for an extra
variable named $z$, which keeps the track of global time. We add a variable
named $s$ of type bitvector at every step which denotes current location. 
Notation $s_\ell$ denotes assertion that current location is $\ell$. 
We also create two binary variables for each channel per automaton -- one for
send and one for receive. 
For example, if automaton $\mathcal{A}_2$ sends over channel $c$ in current
step, we set variable named $\mathcal{A}_2.c!$. 
This notation helps us to identify automaton which uses that channel in the
current step and the type (send or receive) of an operation performed on that
channel. 
We permit to use global clocks. 
While generating formula for $A$, it may happen that some of the automata share
clock names or location names. 
For example, automata  $\mathcal{A}_1$ and $\mathcal{A}_2$ may both have local
clocks named $y$. But we must distinguish between variables that were created 
to hold value of $y$ in $\mathcal{A}_1$ and value of $y$ in $\mathcal{A}_2$. 
We qualify all variables with name of automaton they are 
the part of. Here, we create real variables named $\mathcal{A}_1.y$ and
$\mathcal{A}_2.y$. 
All of these variables are created for every step of a run in a standard
bounded-model-checking fashion. 
Assertions in Fig.~\ref{fig:ta-encoding} 
describe encoding at current and next step in the formula. 
We represent next step variables in primed version. 
For further details refer to \cite{bmc-ta}. 
\begin{figure}
 \begin{equation} \label{eq:ta_switch}
  \bigwedge_{T=(\ell, \varphi, \lambda, \ell')} T \rightarrow \big( s_\ell \wedge \varphi \wedge s'_{\ell'} \wedge
    \bigwedge_{x \in \lambda} (x' = z') \wedge \bigwedge_{x \notin \lambda} (x' = x) \wedge (z' = z)
  \big)
 \end{equation}
 \begin{equation} \label{eq:ta_delay}
  T_\delta \rightarrow \big( (s_\ell = s'_\ell) \wedge (z'-z < 0) \wedge
    \bigwedge_{x \in X} (x' = x) \wedge \bigwedge_{a \in \eta} (\neg a)
  \big)
 \end{equation}
 \begin{equation} \label{eq:ta_null}
  T_{\mbox{null}} \rightarrow \big( (s_\ell = s'_\ell) \wedge (z'=z) \wedge
    \bigwedge_{x \in X} (x' = x) \wedge \bigwedge_{a \in \eta} (\neg a)
  \big)
 \end{equation} 
 \begin{equation} \label{eq:ta_sometrans}
  T_{\mbox{null}} \vee T_\delta \vee \bigvee_{T \in E} T
 \end{equation}  
 \begin{equation} \label{eq:price_init}
 \mbox{price}_0 = 0  
 \end{equation}
\begin{equation} \label{eq:price_switch}
\bigwedge_{T \in E} T \rightarrow \big(\mbox{price}' = \mbox{price} + \psi(T)\big)
 \end{equation} 
 \begin{equation} \label{eq:price_delay}
\bigwedge_{\ell \in L} T_\delta \wedge s_\ell \rightarrow \big(\mbox{price}' = \mbox{price} + \pi(\ell, {\bf x}, z-z')\big)
 \end{equation}
 \begin{equation} \label{eq:price_null}
 T_{\mbox{null}} \rightarrow (\mbox{price'} = \mbox{price})
 \end{equation}
 \begin{equation} \label{eq:price_reach}
 \pricevar^{(n)} \Join k
 \end{equation}
 \label{fig:ta-encoding}
 \vspace{-0.5cm}
 \caption{SMT assertions for priced timed automata}
\end{figure}

\noindent{\bf Extension for Priced Timed Automata.}
Let $\mathcal{A}$ = ($L$, $X$, $E$, $I$, $\pi$, $\psi$) be priced timed
automaton. 
To keep our encoding as general as possible, we  describe our SMT formula
generation for general priced timed automata. 
Observe that this class of automata subsume LPTA, concavely-priced PTA,
piecewise-linear PTA, and Lipschitz-continuous priced PTA. 
For each automaton, we represent current accumulated price using real variable
named  $\pricevar$. 
We introduce variables $\pricevar_k$ at each step.  
Initially $\pricevar_0$ is set to zero as in \ref{eq:price_init}. 
When switch transition occurs, we update the $\pricevar$ using
equation~\ref{eq:price_switch}.   
The function $\psi(T)$ denotes edge price for the transition $T$. 
Equation \ref{eq:price_delay} is used to specify prices for each delay
transition. 
Quantity $(z-z')$ is the delay incurred at  current step and ${\bf x}$ is vector
of current clock valuations. 
As price functions are location dependent, we add clause $s_\ell$ to check
whether current location is $\ell$ and then update price accordingly. 
For null transitions, prices at current and previous step are identical. 
To decide whether accumulated price at step $n$ satisfies the condition
$\pricevar^{(n)} \Join k$, where $\Join \; \in \{ <, \leq, >, \geq, =, \neq\}$,
we add an assertion as per Eq.~\ref{eq:price_reach}. 

\subsection{Experimental Results}
\label{sec:emp}
\begin{table}[t]
 \caption{Comparison of the performance of our tool with UPPAAL-Cora is shown for
   ALP problem with $8$, $9$, and $10$ runways, with varying number of
   airplanes. We report running time (in seconds) for our algorithm (Z3) and
   DFS and random options for UPPAAL-Cora. TO stands for timeout (${>} 30$
   mins). }
 \label{tab:cab}\centering
 \begin{tabular}{|c|c|c|c|c|c|c|c|c|c|c|}
   \hline
   Airplanes & \multicolumn{3}{|c|}{8 runways} & \multicolumn{3}{|c|}{9 runways} & \multicolumn{3}{|c|}{10 runways}\\
   \hline
   & Z3 & CORA & CORA& Z3 & CORA & CORA & Z3 & CORA & CORA\\
   &  & DFS & Random & & DFS & Random & & DFS & Random\\
   \hline
   1 & 0.12 & $<$ 0.1  & $<$ 0.1 & 0.4 & $<$ 0.1 & $<$ 0.1& 0.30 & $<$ 0.1& $<$ 0.1\\
   \hline
   2 & 0.09 & $<$ 0.1 & $<$ 0.1 & 0.57 & $<$ 0.1 & $<$ 0.1 & 0.76 & $<$ 0.1 & $<$ 0.1\\
   \hline
   3 & 0.44 & $<$ 0.1 & $<$ 0.1 & 2.52 & $<$ 0.1 & $<$ 0.1 & 2.31 & $<$ 0.1 &$<$ 0.1\\
   \hline
   4 & 4.28 & 2.4 & 0.04 & 6.73 & 4.18 & 0.08 & 5.86 & 7.81 &0.06\\
   \hline
   5 & 2.73 & 278.21 & 0.7 & 9.61 & 679.27 & 0.1 & 5.09 & TO & 0.05\\
   \hline
   6 & 22.28 & TO & 0.16 & 21.34 & TO & 0.45 & 20.68 & TO & 0.32\\
   \hline
   7 & 29.23 & TO & 0.23 & 201.15 & TO & 1.15& 152.03 & TO & 1.36\\
   \hline
   8 & 89.27 & TO & 0.79 & 86.1 & TO & 1.85 & 94.88 & TO & 5.12\\
   \hline
   9 & 331 & TO & 35.09 & 103.62 & TO & 151.84 & 1650.05 & TO & 277.38\\
   \hline
   10 & 889 & TO & 36.42 & 667.33 & TO & 49.04 & 1309.67 & TO & 230.69\\
   \hline 
 \end{tabular}
\end{table}
We implemented the encoding discussed in the previous subsection as a
vtool~\cite{bhavetool} for analyzing step-bounded optimal-cost for PTA. 
Out tool invokes state-of-the-art theorem prover Z3~\cite{de2008z3} from
Microsoft Research. 
It supports linear and non-linear arithmetic, bit-vectors, arrays, data-types,
and quantifiers.
For our purpose, Z3 can be used to solve price functions that are given as a
polynomial of time-delay and the current valuation.  
Other non-linear price functions such as $\log$, $\sin$, $\cos$, and exp can be
accommodated in this framework using corresponding Taylor series
approximations. 

In order to show experimental results, we concentrate the standard Airport
Landing Scheduling Problem (ALP) from~\cite{behrmann2005optimal}. 
In order to give comparison with an existing tool we keep the price function
linear and compare our tool with state-of-the-art optimal-cost reachability tool
Uppaal-Cora~\cite{cora}. 

\noindent{\bf Airport Landing Scheduling Problem.}
Given number of airplanes each with attributes like type of airplane, landing
time window and number of runways, assign a landing time and runways to each
airplane such that all airplanes land within their specific landing time window
and also comply with  safety regulations like mandatory wake turbulence
separation delay. There are two possible sources of costs. If airplane travels
faster than its designated speed, it lands earlier but consumes more fuel. If
airplane landing is delayed, it suffers fuel costs for circling over the airport.

ALP is known to suffer exponential blowup with increasing runways
\cite{behrmann2005optimal}. 
We used the instances of ALP problem which are distributed with Uppaal CORA demo
version. 
We asked whether there is a schedule such that all airplanes land and total
cost is bounded from above by a fixed budget (800).
Table~\ref{tab:cab} shows the results of our experiments. 
We ran all our experiments on 64-bit Intel$^{\tiny \textregistered}$  Xeon$^{\tiny
  \textregistered}$ CPU E5-2660 v2 running at 2.20GHz with 64 GB RAM. We fixed
time limit to 30 minutes for each problem and used single threaded Z3 SMT solver
(v 4.3.2).

\section{Conclusion and Future Work}
We studied priced timed automata with non-linear prices and showed the
undecidability of a general class of polynomially-priced timed automata.  
We then introduced piecewise-linear and Lipschitz-continuous price functions, and
recovered decidability in this restricted setting. 
We also studied step-bounded cost-optimal reachability problem for price timed
automata, and implemented an SMT based tool to solve this problem. 
This problem is of interest since the optimal-cost reachability problem in
some cases (under structurally non-Zeno restriction on timed automata along 
with non-negativity restriction on prices)
reduces to step-bounded reachability problem. 

Observe that, although our tool does not perform as well as \texttt{random-optimal} option
of UPPAAL-Cora, it outperforms both \texttt{dfs} and \texttt{bfs} (not reported
here). 
As a future work, we plan to exploit randomization to scale the performance of
our implementation. 
We believe that these experiments presented here demonstrate the
applicability of SMT-based step-bounded verification methodology for
medium-sized examples of priced timed automata.

\end{document}